\documentclass[11pt]{article}
\usepackage[T1]{fontenc}
\usepackage[utf8]{inputenc}
\usepackage{authblk}
\usepackage[authoryear]{natbib}
\usepackage{geometry}
\usepackage{ifthen}
\usepackage{graphicx,color,xcolor}
\definecolor{cornellred}{rgb}{0.7, 0.11, 0.11}
\usepackage{hyperref}
\hypersetup{
    colorlinks = true,
    linkcolor=cornellred,
    citecolor=blue,
    linkbordercolor = {white}
}
\usepackage{fancybox}

\usepackage{array,float}
\usepackage{url}
\usepackage{amstext,amssymb,amsmath}
\usepackage{amsthm} 
\usepackage{mathrsfs}
\usepackage{dsfont}
\usepackage[noend]{algorithmic}
\usepackage{algorithm}
\usepackage{pstricks, pst-tree, pst-node}
\usepackage{enumitem}
\usepackage{mathtools}
\usepackage{accents}

\usepackage{romannum}
\AtBeginDocument{\pagenumbering{arabic}}

\theoremstyle{plain}
\newtheorem{proposition}{Proposition}

\newtheorem{fact}{Fact}

\newtheorem{theorem}{Theorem}
\newtheorem{lemma}{Lemma}
\newtheorem{corollary}{Corollary}
\theoremstyle{definition}
\newtheorem{definition}{Definition}

\newtheorem{remark}{Remark}

\newcommand{\Ex}[1]{\mbox{\rm\bf E}\left[#1\right]}

\DeclareMathOperator*{\argmin}{arg\,min}
\DeclareMathOperator*{\argmax}{arg\,max}
\definecolor{ttttff}{rgb}{0.2,0.2,1.}
\definecolor{ttqqqq}{rgb}{0.2,0.,0.}
\definecolor{ududff}{rgb}{0.30196078431372547,0.30196078431372547,1.}
\definecolor{ccqqqq}{rgb}{0.8,0.,0.}

\newcommand{\opt}{\texttt{OPT}}
\renewcommand{\csc}{\texttt{CRSC}}
\newcommand{\HRSC}{\texttt{HRSC}}
\newcommand{\Bcsc}{\texttt{CRBC}}
\title{Hierarchical Clustering with Structural Constraints}
\graphicspath{figures/}
\begin{document}
\title{Hierarchical Clustering with Structural Constraints}
\author[1]{Vaggos Chatziafratis}
\author[1]{Rad Niazadeh}
\author[1]{Moses Charikar}
\affil[1]{Department of Computer Science,  Stanford University.}
\renewcommand\Authands{ and }
\let\theHalgorithm=\thealgorithm
\maketitle


\begin{abstract}

Hierarchical clustering is a popular unsupervised data analysis method. For many real-world applications, we would like to exploit prior information about the data that imposes constraints on the clustering hierarchy, and is not captured by the set of features available to the algorithm. This gives rise to the problem of \textit{hierarchical clustering with structural constraints}. Structural constraints pose major challenges for 
bottom-up approaches like average/single linkage and even though they can be naturally incorporated into top-down divisive algorithms, no formal guarantees exist on the quality of their output. In this paper, we provide provable approximation guarantees for two simple top-down algorithms, 
using a recently introduced optimization viewpoint of hierarchical clustering with pairwise similarity information~\citep{dasguptaSTOC}. We show how to find good solutions even in the presence of conflicting prior information, by formulating a \emph{constraint-based regularization} of the objective. We further explore a variation of this objective for dissimilarity information~\citep{vincentSODA} and improve upon current techniques. Finally, we demonstrate our approach on a real dataset for the taxonomy application.
\end{abstract}

\section{Introduction}

Hierarchical clustering (HC) is a widely used data analysis tool, ubiquitous in information retrieval, data mining, and machine learning~(see a survey by \cite{berkhin2006survey}). This clustering technique represents a given dataset as a binary tree; each leaf represents an individual data point and each internal node represents a cluster on the leaves of its descendants.
HC has become the most popular method for gene expression data analysis~\cite{eisen1998cluster}, and also has been used in the analysis of social networks~\cite{leskovec2014mining,mann2008use}, bioinformatics~\cite{diez2015novel}, image and text classification~\cite{steinbach2000comparison}, and even in analysis of financial markets~\cite{tumminello2010correlation}.
It is attractive because it provides richer information at all levels of granularity simultaneously, compared to more traditional \textit{flat} clustering approaches like $k$-means or $k$-median. 

Recently, \cite{dasguptaSTOC} formulated HC as a combinatorial optimization problem, giving a principled way to compare the performance of different HC algorithms.
This optimization viewpoint has since received a lot of attention~\cite{royNIPS,vaggosSODA,vincentNIPS,joshNIPS,vincentSODA} that has led not only to the development of new algorithms but also to theoretical justifications for the observed success of popular HC algorithms (e.g. average-linkage).


However, in real applications of clustering, the user often has \textit{background knowledge} about the data that may not be captured by the input to the clustering algorithm.
There is a rich body of work on constrained (flat) clustering formulations that take into account such user input in the form of ``cannot link'' and ``must link'' constraints~\cite{wagstaff2000clustering,wagstaff2001constrained, bilenko2004integrating, AISTATS}.
Very recently, ``semi-supervised'' versions of HC that incorporate additional constraints have been studied~\cite{dasguptaICML}, where the natural form of such constraints is triplet (or ``must link before'') constraints $ab|c$\footnote{Hierarchies on data imply that all datapoints are linked at the highest level and all are separated at the lowest level, hence ``cannot link'' and ``must link'' constraints are not directly meaningful.}: these require that valid solutions contain a sub-cluster with $a,b$ together and $c$ previously separated from them.\footnote{For a concrete example from taxonomy of species, a triplet constraint may look like ($\textsc{Tuna}, \textsc{Salmon}|\textsc{Lion}$).} 
Such triplet constraints, as we formally show later, can encode more general structural constraints in the form of rooted subtrees.
Surprisingly, such simple triplet constraints already pose significant challenges for bottom-up linkage methods.
(\hyperref[fig:fail]{Figure~\ref{fig:fail}}).
\begin{figure}[ht]

\begin{center}
\centerline{\includegraphics[scale=0.62]{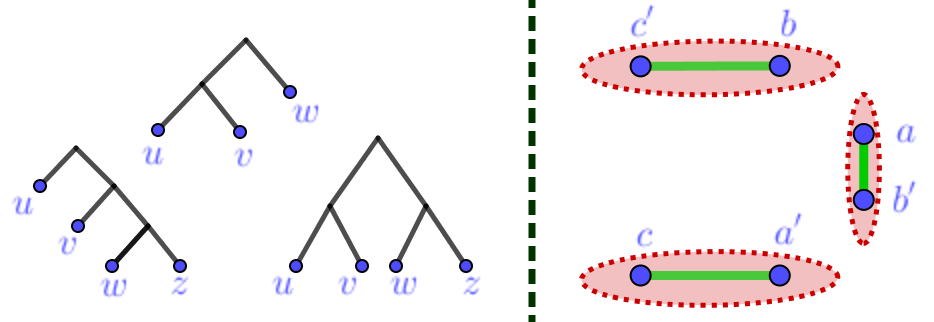}} \label{fig:fail}
\caption{ \emph{(Left)} Example of a triplet constraint $uv|w$ and more general rooted tree constraints on 4 points $u,v,w,z$. \emph{(Right)} Example with only two constraints $ab|c,a'b'|c'$ demonstrating that popular distance-based linkage algorithms may fail to produce valid HC. Here they get stuck after 3 merging steps (green edges).}
\end{center}
\end{figure}

Our work is motivated by applying the optimization lens to study the interaction of hierarchical clustering algorithms with structural constraints. Constraints can be fairly naturally incorporated into top-down (i.e. divisive) algorithms for hierarchical clustering; but can we establish guarantees on the quality of the solution they produce? Another issue is that incorporating constraints from multiple experts may lead to a conflicting set of constraints; can the optimization viewpoint of hierarchical clustering still help us obtain good solutions even in the presence of infeasible constraints? Finally, different objective functions for HC have been studied in the literature; do algorithms designed for these objectives behave similarly in the presence of constraints? To the best of our knowledge, this is the first work to propose a unified approach for constrained HC through the lens of optimization and to give provable approximation guarantees for a collection of fast and simple top-down algorithms that have been used for unconstrained HC in practice (e.g. community detection in social networks~\cite{mann2008use}).

\paragraph{Background on Optimization View of HC.}


\cite{dasguptaSTOC} introduced a natural optimization framework for HC. Given a weighted graph $G(V,E,w)$ and pairwise
similarities $w_{ij}\ge 0$ between the $n$ data points $i,j\in V$, the goal is to find a hierarchical tree $T^*$ such that
\begin{align}
T^*=\argmin_{\textrm{all trees}~T}~\sum_{(i,j)\in E}w_{ij}\cdot \lvert T_{ij}\rvert \label{obj}
\end{align}
where $T_{ij}$ is the subtree rooted at the lowest common ancestor of $i,j$ in $T$ and $\lvert T_{ij}\rvert$ is the number of leaves it contains.\footnote{Observe that in HC, all edges get cut eventually. Therefore it is better to postpone cutting ``heavy'' edges to when the clusters become small, i.e .as far down the tree as possible.} We denote \eqref{obj} as \emph{similarity-HC}. For applications where the geometry of the data is given by dissimilarities,  again denoted by $\{w_{ij}\}_{(i,j)\in E}$, ~\cite{vincentSODA} proposed an analogous approach, where the goal is to find a hierarchical tree $T^*$ such that
 \begin{align}
T^*=\argmax_{\textrm{all trees}~T}~\sum_{(i,j)\in E}w_{ij}\cdot \lvert T_{ij}\rvert \label{obj-3}
\end{align} 
We denote \eqref{obj-3} as \emph{dissimilarity-HC}. A comprehensive list of desirable properties of the aformentioned objectives can be found in~\cite{dasguptaSTOC,vincentSODA}. In particular, if there is an underlying ground-truth hierarchical structure in the data, then $T^*$ can recover the ground-truth. Also, both objectives are \textsc{NP}-hard to optimize, so the focus is on approximation algorithms.
\paragraph{Our Results.}
 \emph{\romannum{1}}) We design algorithms that take into account both the geometry of the data, in the form of similarities, and the structural constraints imposed by the users. Our algorithms emerge as the natural extensions of Dasgupta's original recursive sparsest cut algorithm and the recursive balanced cut suggested in \cite{vaggosSODA}. We generalize previous analyses to handle constraints and we prove an $O(k\alpha_n)$-approximation guarantee\footnote{For $n$ data points, $\alpha_n=O(\sqrt{\log n})$ is the best approximation factor for the sparsest cut and $k$ is the number of constraints.}, thus surprisingly matching the best approximation guarantee of the \emph{unconstrained} HC problem for constantly many constraints.
 
\emph{\romannum{2}}) In the case of infeasible constraints, we extend the similarity-HC optimization framework, and we measure the quality of a possible tree $T$ by a \textit{constraint-based regularized} objective. The regularization naturally favors solutions with as few constraint violations as possible and as far down the tree as possible (similar to the motivation behind similarity-HC objective). For this problem, we provide a top-down $O(k\alpha_n)$-approximation algorithm by drawing an interesting connection to an instance of the hypergraph sparsest cut problem. 

\emph{\romannum{3}})  We then change gears and study the dissimilarity-HC objective. Surprisingly, we show that known top-down techniques do not cope well with constraints, drawing a contrast with the situation for similarity-HC. Specifically, the \emph{(locally) densest cut heuristic} performs poorly even if there is only one triplet constraint, blowing up its approximation factor to $O(n)$. 
Moreover, we improve upon the state-of-the-art in~\cite{vincentSODA}, by showing a simple randomized partitioning is a $\tfrac23$-approximation algorithm. We also give a deterministic \textit{local-search} algorithm with the same worst-case guarantee. Furthermore, we show that our randomized algorithm is robust under constraints, mainly because of its ``exploration'' behavior. In fact, besides the number of constraints, we propose an inherent notion of \emph{dependency measure} among constraints to capture this behavior quantitatively. This helps us not only to explain why ``non-exploring'' algorithms may perform poorly, but also gives tight guarantees for our randomized algorithm.

\paragraph{Experimental Results.} We run experiments on the Zoo dataset~\citep{Zoo} to demonstrate our approach and the performance of our algorithms for a taxonomy application. We consider a setup where there is a ground-truth tree and extra information regarding this tree is provided for the algorithm in the form of triplet constraints. The upshot is we believe specific variations of our algorithms can exploit this information; In this practical application, our algorithms have around $\%9$ imrpvements in the objective compared to the naive recursive sparsest cut proposed in \cite{dasguptaSTOC} that does not use this information. See Appendix~\ref{sec:appendix-zoo} for more details on the setup and precise conclusions of our experiments. 

\paragraph{Constrained HC work-flow in Practice.} Throughout this paper, we develop different tools to handle user-defined structural constraints for hierarchical clustering. Here we describe a recipe on how to use our framework in practice.

\textit{(1) Preprocessing constraints to form triplets.} User-defined structural constraints as rooted binary subtrees are convenient for the user and hence for the usability of our algorithm. The following proposition (whose proof is in the supplement) allows us to focus on studying HC with just triplet constraints.

\begin{proposition}
\label{prop:convert}
Given constraints as a rooted binary subtree $T$ on $k$ data points ($k\ge3$), there is linear time algorithm that returns an equivalent set of at most $k$ triplet constraints.
\end{proposition}

\textit{(2) Detecting feasibility.} The next step is to see if the set of triplet constraints is consistent, i.e. whether there exists a HC satisfying all the constraints. For this, we use a simple linear time algorithm called \texttt{BUILD}~\cite{ahoBUILD}. 

\textit{(3) Hard constraints vs. regularization.} \texttt{BUILD} can create a hierarchical decomposition that satisfies triplet constraints, but ignores the \textit{geometry} of the data, whereas our goal here is to consider both simultaneously. Moreover, in the case that the constraints are infeasible, we aim to output a clustering that minimizes the cost of violating constraints combined with the cost of the clustering itself.

$\mathbf{\bullet}$~\textit{Feasible instance:} to output a feasible HC, we propose using \emph{Constrained Recursive Sparsest Cut} (\csc) or  \emph{Constrained Recursive Balanced Cut} (\Bcsc): two simple top-down algorithms which are natural adaptations of recursive sparsest cut~\citep{mann2008use, dasguptaSTOC} or recursive balanced cut~\cite{vaggosSODA} to respect constraints \hyperref[sec:CSC]{(Section~\ref{sec:CSC})}. 

$\mathbf{\bullet}$~\textit{Infeasible instance:} in this case, we turn our attention to a regularized version of HC, where the cost of violating constraints is added to the tree cost. We then propose an adaptation of \csc, namely \emph{Hypergraph Recursive Sparsest Cut} (\HRSC) for the regularized problem \hyperref[sec:HRSC]{(Section~\ref{sec:HRSC})}.

\paragraph{Real-world application example.} In phylogenetics, which is the study of the evolutionary history and relationships among species, an end-user usually has access to whole genomes data of a group of organisms. There are established methods in phylogeny to infer similarity scores between pairs of datapoints, which give the user the similarity weights $w_{ij}$. Often the user also has access to rare structural footprints of a common ancestry tree (e.g. through gene rearrangement data, gene inversions/transpositions etc., see~\cite{phylogenomics}). These rare, yet informative, footprints play the role of the structural constraints. The user can follow our pre-processing step to get triplet constraints from the given rare footprints, and then use Aho’s BUILD algorithm to choose between regularized or hard version of the HC problem. The above illustrates how to use our workflow and why using our algorithms facilitates HC when expert domain knowledge is available.

\paragraph{Further related work.} Similar to~\cite{dasguptaICML}, constraints in the form of triplet queries have been used in an (adaptive) active learning framework by~\cite{tamuzICML,kempeSODA}, showing that approximately $O(n\log n)$ triplet queries are enough to learn an underlying HC. Other forms of user interaction in order to improve the quality of the produced clusterings have been used in~\cite{nina1,nina2} where they prove that interactive feedback in the form of cluster split/merge requests can lead to significant improvements. Robust algorithms for HC in the presence of noise were studied in~\cite{nina3} and a variety of sufficient conditions on the similarity function that would allow linkage-style methods to produce good clusters was explored in~\cite{nina4}. On a different setting, the notion of triplets has been used as a measure of \textit{distance} between hierarchical decomposition trees on the same data points~\cite{brodal2013efficient}. More technically distant analogs of how to use relations among triplets points have recently been proposed in~\cite{kernelNIPS} for defining kernel functions corresponding to high-dimensional embeddings.


\section{Constrained Sparsest (Balanced) Cut}
\label{sec:CSC}

Given an instance of the constrained hierarchical clustering,  our proposed \csc~algorithm uses a blackbox $\alpha_n$-approximation algorithm for the sparsest cut problem (the best-known approximation factor for this problem is $O(\sqrt{\log n})$ due to~\cite{arora2009expander}). Moreover, it also maintains the feasibility of the solution in a top-down approach by recursive partitioning of what we call the \textit{supergraph} $G'$. Informally speaking, the supergraph is a simple data structure to track the progress of the algorithm and the resolved constraints. 

More formally,  for every constraint $ab|c$ we merge the nodes $a$ and $b$ into a \textit{supernode} $\{a,b\}$ while maintaining the edges in $G$ (now connecting to their corresponding supernodes). Note that $G'$ may have parallel edges, but this can easily be handled by grouping edges together and replacing them with the sum of their weights. We repeatedly continue this \textit{merging} procedure until there are no more constraints. Observe that any feasible solution needs to start splitting the original graph $G$ by using a cut that is also present in $G'$.  When cutting the graph $G'=(G_1,G_2)$, if a constraint $ab|c$ is \textit{resolved},\footnote{A constraint $ab|c$ is \textit{resolved}, if $c$ gets separated from $a,b$.} then we can safely \textit{unpack} the supernode $\{a,b\}$ into two nodes again (unless there is another constraint $ab|c'$ in which case we should keep the supernode). By continuing and recursively finding approximate sparsest cuts on the supergraph $G_1$ and $G_2$, we can find a feasible hierarchical decomposition of $G$ respecting all triplet constraints. Next, we show the approximation guarantees for our algorithm.

\begin{algorithm}[ht]
\caption{\texttt{CRSC}}
\label{alg: Constrained Sparsest Cut}
\begin{algorithmic}[1]

\STATE Given $G$ and the triplet constraints $ab|c$, run \texttt{BUILD} to create the supergraph $G'$.
\STATE Use a blackbox access to an $\alpha_n$-approximation oracle for the sparsest cut problem, i.e. ${\argmin}_{S\subseteq V}\frac{w_{G'}(S,\bar{S})}{\lvert S\rvert \cdot\lvert \bar{S}\rvert}$.
\STATE Given the output cut $(S,\bar{S})$, separate the graph $G'$ into two pieces $G_1(S,E_1)$ and $G_2(V\setminus S,E_2)$. 
\STATE Recursively compute a HC $T_1$ for $G_1$ using only $G_1$'s active constraints. Similarly compute $T_2$ for $G_2$.
\STATE Output $T=(T_1,T_2)$.
\end{algorithmic}
\end{algorithm}
 \vspace{-0.3cm}
\paragraph{Analysis of \csc~Algorithm.} The main result of this section is the following theorem:

\begin{theorem}\label{th:main}
Given a weighted graph $G(V,E,w)$ with $k$ triplet constraints $ab|c$ for $a,b,c\in V$, the \csc\ algorithm outputs a HC respecting \textit{all} triplet constraints and achieves an
$O(k\alpha_n)$-approximation for the HC-similarity objective as in \hyperref[obj]{(\ref{obj})}.
\end{theorem}


\paragraph{Notations and Definitions.} We slightly abuse notation by having $\opt$ denote the optimum hierarchical decomposition or its optimum value as measured by \hyperref[obj]{(\ref{obj})}. Similarly for \csc. For $t\in [n]$, \opt$(t)$ denotes the maximal clusters in \texttt{OPT} of size at most $t$. Note that \texttt{OPT}$(t)$ induces a partitioning of $V$. We use \texttt{OPT}$(t)$ to denote edges cut by \texttt{OPT}$(t)$ (i.e. edges with endpoints in different clusters in \texttt{OPT}$(t)$) or their total weight; the meaning will be clear from context. For convenience, we define \texttt{OPT}$(0)$ $=\sum_{(i,j)\in E}w_{ij}$. For a cluster $A$ created by \csc, a constraint $ab|c$ is \textit{active} if $a,b,c\in A$, otherwise $ab|c$ is resolved and can be discarded.

\paragraph{Overview of the Analysis. } There are three main ingredients: The first is to view a HC of $n$ datapoints as a collection of partitions, one for each level $t = n-1,\dots, 1$, as in \citep{vaggosSODA}. For a level $t$, the partition consists of maximal clusters of size at most $t$. The total cost incurred by $\opt$ is then a combination of costs incurred at each level of this partition. This is useful for comparing our \csc\ cost with $\opt$. The second idea is in handling constraints and it is the main obstacle where previous analyses~\cite{vaggosSODA,vincentSODA} break down: constraints inevitably limit the possible cuts that are feasible at any level, and since the set of active constraints\footnote{All constraints are \textit{active} in the beginning of \csc.} differ for $\csc$ and \opt, a direct comparison between them is impossible. If we have no constraints, we can charge the cost of partitioning a cluster $A$ to lower levels of the $\opt$ decomposition. However, when we have triplet constraints, the partition induced by the lower levels of $\opt$ in a cluster $A$ \emph{will not be feasible in general} (\hyperref[fig:clusterA]{Figure~\ref{fig:clusterA}}). The natural way to overcome this obstacle is merging pieces of this partition so as to respect constraints and using higher levels of $\opt$, but it still may be impossible to compare \csc\ with \opt\ if all pieces are merged. We overcome this difficulty by an indirect comparison between the \csc\ cost and lower levels $\tfrac{r}{6k_A}$ of \opt, where $k_A$ is the number of active constraints in $A$. Finally, after a cluster-by-cluster analysis bounding the \csc\ cost for each cluster, we exploit disjointness of clusters of the same level in the \csc\ partition allowing us to combine their costs.
\begin{proof}[Proof of Theorem~\ref{th:main}]
We start by borrowing the following facts from \citep{vaggosSODA}, modified slightly for the purpose of our analysis (proofs are provided in the supplementary materials).
\begin{fact}[\textbf{Decomposition of \opt}]
\label{fact:dec}
The total cost paid by $\opt$ can be decomposed into costs of the different levels in the \opt\ partition, i.e. $\opt=\sum_{t=0}^n w({\opt(t))}$.
\end{fact}
\begin{fact}[\textbf{\opt~at scaled levels}]
\label{fact:scale}
Let $k\le \tfrac{n}{6}$ be the number of constraints. Then, $\opt\ge\tfrac{1}{6k}\cdot \sum_{t=0}^nw({\opt(\lfloor\tfrac{t}{6k}\rfloor))}$.
\end{fact}
\begin{figure}[ht]
\begin{center}
\centerline{\includegraphics[scale=0.6]{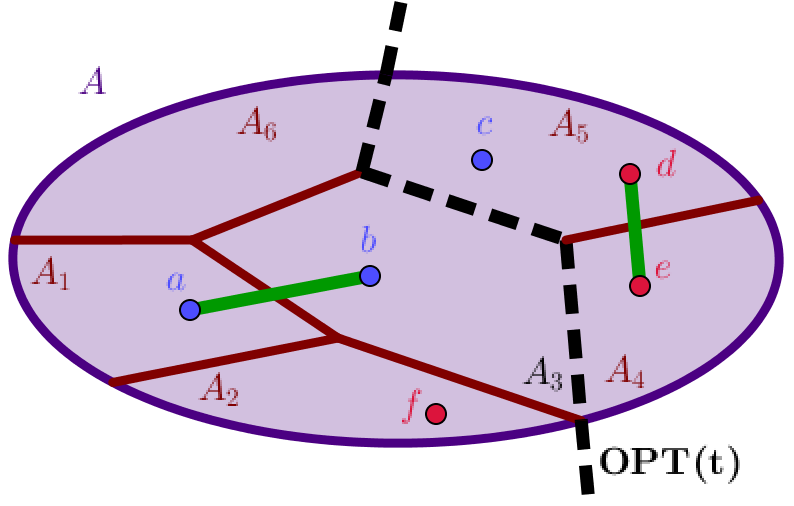}} 
\caption{The main obstacle in the constrained HC problem is that our algorithm has different \textit{active constraints} compared to \texttt{OPT}. Both $ab|c,de|f$ constraints are resolved by the cut \texttt{OPT(t)}. \label{fig:clusterA}}
\end{center}
\vspace{-0.6cm}
\end{figure}
Given the above facts, we look at any cluster $A$ of size $r$ produced by the algorithm. Here is the main technical lemma that allows us to bound the cost of $\csc$ for partitioning $A$.
\begin{lemma}\label{lem:mainproof}
Suppose $\csc$ partitions a cluster $A$ ($|A|=r$) in two clusters $(B_1,B_2)$ (w.l.o.g. $|B_1|=s,|B_2|=r-s,s\le \lfloor \tfrac r2\rfloor \le r-s$). Let the size $r\ge 6k$ and let $l=6k_A$, where $k_A$ denotes the number of \textit{active constraints} for $A$. Then: $r\cdot w(B_1,B_2)\le 4\alpha_n \cdot s\cdot w(\opt(\lfloor\tfrac{r}{l}\rfloor)\cap A)$.
\end{lemma}
\begin{proof}
The cost incurred by $\csc$ for partitioning $A$ is $r\cdot w(B_1,B_2)$. Now consider $\opt(\lfloor\tfrac{r}{l}\rfloor)$. This induces a partitioning of  $A$ into pieces $\{A_i\}_{i\in[m]}$, where by design  $|A_i|=\gamma_i |A|$, $\gamma_i\le \tfrac 1l, \forall i \in [m]$. Now, consider the cuts $\{(A_i,A\setminus A_i)\}$. Even though all $m$ cuts are allowed for $\opt$, for $\csc$ some of them are forbidden: for example, in \hyperref[fig:clusterA]{Figure~\ref{fig:clusterA}}, the constraints $ab|c,de|f$ render 4 out of the 6 cuts infeasible. But how many of them can become infeasible with $k_A$ active constraints? Since every constraint is involved in at most 2 cuts, we may have at most $2k_A$ infeasible cuts. Let $F\subseteq [m]$ denote the index set of feasible cuts, i.e. if $i\in F$, the cut $(A_i,A\setminus A_i)$ is feasible for $\csc$. To cut $A$, we use an $\alpha_n$-approximation of sparsest cut, whose sparsity is upper bounded by \textit{any feasible} cut:
\begin{align*}
\frac{w(B_1,B_2)}{s(r-s)}&\le \alpha_n \cdot \textrm{SP.CUT}(A)\le\alpha_n \min_{i\in F}\frac{w(A_i,A\setminus A_i)}{|A_i||A\setminus A_i|}
\le  \alpha_n \frac{\sum_{i\in F} w(A_i,A\setminus A_i)}{\sum_{i\in F} |A_i||A\setminus A_i|}
\end{align*}
where for the last inequality we used the standard fact that $\min_{i}\tfrac{\mu_i}{\nu_i}\le\tfrac{\sum_i\mu_i}{\sum_i\nu_i}$ for $\mu_i\ge0$ and $\nu_i>0$. We also have the following series of inequalities:
\begin{align*}
\alpha_n \frac{\sum_{i\in F} w(A_i,A\setminus A_i)}{\sum_{i\in F} |A_i||A\setminus A_i|}&\le \alpha_n\frac{2w(\opt(\lfloor\tfrac rl\rfloor)\cap A)}{r^2\sum_{i\in F}\gamma_i(1-\gamma_i)}
\le 4\alpha_n\frac{w(\opt(\lfloor\tfrac rl\rfloor)\cap A)}{r^2}
\end{align*}
where the first inequality holds because we double count some (potentially all) edges of $\opt(\lfloor\tfrac rl\rfloor)\cap A$ (these are the edges cut by $\opt(\lfloor\tfrac rl\rfloor)$ that are also present in cluster $A$, i.e. they have both endpoints in $A$) and the second inequality holds because $\gamma_i\le \frac{1}{6k}\implies 1-\gamma_i\ge \frac{6k-1}{6k}$ and 
\begin{align*}
\sum_{i\in F} \gamma_i(1-\gamma_i)&\ge \sum_{i=1}^m\gamma_i(1-\gamma_i)-2\sum_{i\in [m]\setminus F}\frac{1}{6k}\ge \frac{6k-1}{6k}\sum_{i=1}^m\gamma_i-\frac{2k}{6k}=\frac{4k-1}{6k}\ge 1/2~
\end{align*}
Finally, we are ready to prove the lemma by combining the above inequalities ($\tfrac{r-s}{r}\le1$):
\begin{align*}
r\cdot w(B_1,B_2)&=r\cdot s(r-s)\cdot \frac{w(B_1,B_2)}{s(r-s)}\\
&\le r\cdot s(r-s)\cdot 4\alpha_n\frac{w(\opt(\lfloor\tfrac rl\rfloor)\cap A)}{r^2}\le 4\alpha_n \cdot s\cdot w(\opt(\lfloor\tfrac{r}{l}\rfloor)\cap A).\qedhere
\end{align*}
\end{proof}
It is clear that we exploited the charging to lower levels of \opt, since otherwise if all pieces in $A$ were merged, the denominator with the $|A_i|$'s would become 0. The next lemma lets us combine the costs incurred by $\csc$ for different clusters $A$ (proof is in the supplementary materials)
\begin{lemma}[\textbf{Combining the costs of clusters in \csc}]
\label{lem:combine}
The total \csc\ cost for partitioning all clusters $A$ into $(B_1,B_2)$ (with $|A|=r_A,|B_1|=s_A$) is bounded by:
\begin{align}
&(1)~~\sum_{A:\lvert A\rvert \geq 6k} r_A\cdot w(B_1,B_2)\le O(\alpha_n) \cdot\sum_{t=0}^nw(\opt(\lfloor\tfrac{t}{6k}\rfloor))\nonumber\\
&(2)~~\sum_{A: \lvert A\rvert<6k}r_Aw(B_1,B_2)\le 6k\cdot \opt\nonumber
\end{align}
\end{lemma}
Combining   \hyperref[fact:scale]{Fact~\ref{fact:scale}} and  \hyperref[lem:combine]{Lemma~\ref{lem:combine}} finishes the proof. 
\end{proof}

\begin{remark} In the supplementary material, we prove how one can use \emph{balanced cut}, i.e. finding a cut $S$ such that
\begin{equation}
\underset{S\subseteq V:\lvert S\rvert \geq n/3,\lvert\bar{S}\rvert\geq n/3}{\argmin}w_{G'}(S,\bar{S})
\end{equation}
instead of sparsest cut, and using approximation algorithms for this problem achieves the same approximation factor as in \hyperref[th:main]{Theorem~\ref{th:main}}, but with better running time. 
\end{remark}

\begin{remark}
Optimality of the CRSC algorithm: Note that complexity theoretic lower-bounds for the unconstrained version of HC from~\cite{vaggosSODA} also apply to our setting; more specifically, they show that no constant factor approximation exists for HC assuming the Small-Set Expansion Hypothesis.
\end{remark}

\begin{theorem}[The divisive algorithm using balanced cut] 
Given a weighted graph $G(V,E,w)$ with $k$ triplet constraints $ab|c$ for $a,b,c\in V$, the constrained recursive balanced cut algorithm~\Bcsc~(same as $\csc$, but using balanced cut instead of sparsest cut) outputs a HC respecting \textit{all} triplet constraints and achieves an $O(k\alpha_n)$-approximation for Dasgupta's HC objective. Moreover, the running time is almost linear time.
\end{theorem}

\section{Constraints and Regularization}
\label{sec:HRSC}
Previously, we assumed that constraints were feasible. However, in many practical applications, users/experts may disagree, hence our algorithm may receive conflicting constraints as input. Here we want to explore how to still output a satisfying HC that is a good in terms of objective\hyperref[obj]{~(\ref{obj})} (similarity-HC) and also respects the constraints as much as possible. To this end, we propose a \emph{regularized} version of Dasgupta's objective, where the regularizer measures quantitatively the degree by which constraints get violated.

Informally, the idea is to penalize a constraint more if it is violated at top levels of the decomposition compared to lower levels. We also allow having different violation \emph{weights} for different constraints (potentially depending on the expertise of the users providing the constraints). More concretely, inspired by the Dasgupta's original objective function, we consider the following optimization problem: 
\begin{align}
\min_{T\in \mathcal{T}}\bigg{(}&\sum_{(i,j)\in E}w_{ij}|T_{ij}|+\lambda\cdot\sum_{ab|c \in \mathcal{K}}c_{ab|c}|T_{ab}|\cdot \mathbf{1}\{ab|c \text{\ is violated}\}\bigg{)}, \label{obj2}
\end{align} 
where $\mathcal{T}$ is the set of all possible binary HC trees for  the given data points, $\mathcal{K}$ is the set of the $k$ triplet constraints, $T_{ab}$ is the size of the subtree rooted at the least common ancestor of $a,b$,  and $c_{ab|c}$ is defined as the \emph{base cost} of violating triplet constraint $ab|c$. Note that the regularization parameter $\lambda\ge0$ allows us to interpolate between satisfying the constraints or respecting the geometry of the data. 



\paragraph{Hypergraph Recursive Sparsest Cut} In order to design approximation algorithms for the regularized objective, 
we draw an interesting connection to a different problem, which we call \emph{3-Hypergraph Hierarchical Clustering (3HHC)}. An instance of this problem consists of a hypergraph $G^{\mathcal{H}}=(V,E,E^{\mathcal{H}})$ with edges $E$, and hyperedges of size $3$, $E^\mathcal{H}$, together with similarity weights for edges, $\{w_{ij}\}_{(i,j)\in E}$, and similarity weights for 3-hyperedges,\footnote{We have 3 different weights corresponding to the 3 possible ways of partitioning $\{i,j,k\}$ in two parts: $w_{ij|k}$, $w_{jk|i}$ and $w_{ki|j}$. } $\{w_{ij|k}\}_{(i,j,k)\in E^{\mathcal{H}}}$. We now think of HC on the hypergraph $G^{\mathcal{H}}$, where for every binary tree $T$ we define the cost to be the natural extension of Dasgupta's objective: 
\begin{align}
\label{eq:HHC}
\sum_{(i,j)\in E}w_{ij}\lvert T_{ij}\rvert +\sum_{(i,j,k)\in E^{\mathcal{H}}}w^T_{ijk}\lvert T_{ijk}\rvert 
\end{align}
where $w^T_{ijk}$ is either equal to $w_{ij|k}, w_{jk|i}~\textrm{or}~w_{ki|j}$, and $T_{ijk}$ is either the subtree rooted at  LCA($i,j$),\footnote{LCA($i,j$) denotes the lowest common ancestor of $i,j \in T$.} LCA($i,k$) or LCA($k,j$), all depending on how $T$ cuts the 3-hyperedge $\{i,j,k\}$.
The goal is to find a hierarchical clustering of this hypergraph, so as to minimize the cost \hyperref[eq:HHC]{(\ref{eq:HHC})} of the tree.

\paragraph{Reduction from Regularization to 3HHC.} Given an instance of HC with constraints (with their costs of violations) and a parameter $\lambda$, we create a hypergraph $G^\mathcal{H}$ so that the total cost of any binary clustering tree in the 3HHC problem \hyperref[eq:HHC]{(\ref{eq:HHC})} corresponds to the regularized objective of the same tree as in \hyperref[obj2]{(\ref{obj2})}. $G^{\mathcal{H}}$ has exactly the same set of vertices, (normal) edges and (normal) edge weights as in the original instance of the HC problem. Moreover, for every constraint $ab|c$ (with cost $c_{ab|c}$) it has a hyperedge $\{a,b,c\}$, to which we assign three weights $w_{ab|c}=0,w_{ac|b}= w_{bc|a}=\lambda \cdot c_{ab|c}$. Therefore, we ensure that any divisive algorithm for the 3HHC problem avoids the cost $\lvert T_{abc}\rvert\cdot \lambda \cdot c_{ab|c}$ only if it chops $\{a,b,c\}$ into $\{a,b\}$ and $\{c\}$ at some level, which matches the regularized objective.

\paragraph{Reduction from 3HHC to Hypergraph Sparsest Cut.} 
A natural generalization of the sparsest cut problem for our hypergraphs, which we call \emph{Hyper Sparsest Cut (HSC)}, is the following problem:
\begin{align*}
 \argmin_{S\subseteq V} \left(\frac{w(S,\bar{S})+\sum_{(i,j,k)\in E^{\mathcal{H}}}{w^{S}_{ijk}}}{\lvert S\rvert \lvert \bar{S}\rvert} \right)~,
\end{align*}
where $w(S,\bar{S})$ is the weight of the cut $(S,\bar{S})$ and $w^{S}_{ijk}$ is either equal to $w_{ij|k},w_{jk|i}$ or $w_{ki|j}$, depending on how $(S,\bar{S})$ chops the hyperedge $\{i,j,k\}$. Now, similar to \cite{vaggosSODA,dasguptaSTOC}, we can recursively run a blackbox approximation algorithm for HSC to solve 3HHC. The main result of this section is the following technical proposition, whose proof is analogous to that of \hyperref[th:main]{Theorem~\ref{th:main}} (provided in the supplementary materials).

\begin{proposition}
\label{prop:HSC}
Given the hypergraph $G^{\mathcal{H}}$ with $k$ hyperedges, and given access to an algorithm which is $\alpha_n$-approximation for HSC,  the Recursive Hypergraph Sparsest Cut (\texttt{R-HSC}) algorithm achieves an
$O(k\alpha_n)$-approximation.
\end{proposition}

\paragraph{Reduction from HSC back to Sparsest Cut.}




We now show how to get an $\alpha_n$-approximation oracle for our instance of the HSC problem by a general reduction to sparsest cut. Our reduction is simple: given a hypergraph $G^\mathcal{H}$ and all the weights, create an instance of  sparsest cut with the same vertices, (normal) edges and (normal) edge weights. Moreover, for every 3-hyperedge $\{a,b,c\}$ consider adding a \emph{triangle} to the graph, i.e. three weighted edges connecting $\{a,b,c\}$, where:
\begin{align*}
&w'_{ab}=\frac{w_{bc|a}+w_{ac|b}-w_{ab|c}}{2}=\lambda\cdot c_{ab|c}, \\
&w'_{ac}=\frac{w_{bc|a}+w_{ab|c}-w_{ac|b}}{2}=0,\\
&w'_{bc}=\frac{w_{ac|b}+w_{ab|c}-w_{bc|a}}{2}=0.
\end{align*}
This construction can be seen in \hyperref[fig:reduction]{Figure~\ref{fig:reduction}}. The important observation is that $w'_{ab}+w'_{ac}=w_{bc|a}$, $w'_{ab}+w'_{bc}=w_{ac|b}$ and $w'_{bc}+w'_{ac}=w_{ab|c}$, which are exactly the weights associated with the corresponding splits of the 3-hyperedge $\{a,b,c\}$. So, correctness of the reduction\footnote{Since all weights in the final graph are non-negative, standard techniques for Sparsest Cut can be used.} follows as the weight of each cut is preserved between the hypergraph and the graph after adding the triangles. For a discussion on extending this gadget more generally, see the supplement.

\begin{remark}
Reduction to hypergraphs: we would like to emphasize the necessity of the hypergraph version in order for the reduction to work. One might think that just adding extra heavy edges would be sufficient, but there is a technical difficulty with this approach. Consider a triplet constraint $ab|c$; once $c$ is separated from $a$ and $b$ at some level, there is no extra tendency anymore to keep $a$ and $b$ together (i.e. only the similarity weight should play role after this point). This behavior cannot be captured by only adding heavy-weight edges. Instead, one needs to add a heavy edge between $a$ and $b$ that disappears once $c$ is separated, and this is exactly why we need the hyperedge gadget. One can replace the reduction for a one-shot proof, but we believe it will be less modular and less transparent.
\end{remark}

\begin{figure}[ht]
\vskip -0.1in
\begin{center}
\centerline{\includegraphics[scale=0.9]{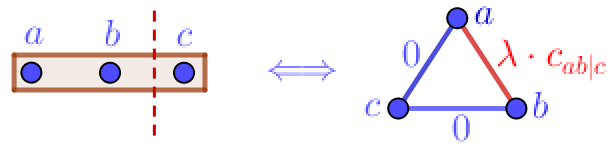}} 
\caption{Transforming a 3-hyperedge to a triangle.\label{fig:reduction}}
\end{center}
\vskip -0.3in
\end{figure}

\newcommand{\Ycal}{\mathcal{Y}}
\newcommand{\Ecal}{\mathcal{E}}
\newcommand{\Zcal}{\mathcal{Z}}
\newcommand{\OPT}{\texttt{OPT}}
\newcommand{\SDP}{\texttt{OPT-SDP}}
\newcommand{\Ws}{\sum_{(i,j)\in E}w_{ij}}
\newcommand{\OBJ}{\texttt{OBJ}}
\newcommand{\RRC}{\texttt{Recursive-Random-Cutting}~}
\newcommand{\cons}{\mathcal{C}}
\newcommand{\layer}{\mathcal{I}}
\newcommand{\dm}[1]{\textrm{DM}(#1)}
\newcommand{\dmc}[1]{\textrm{DMC}(#1)}
\newcommand{\CRRC}{\texttt{Constrained-RRC}~}

\section{Variations on a Theme}
\label{sec:random}
In this section we study dissimilarity-HC, and we look into the problem of designing approximation algorithms for both unconstrained and constrained hierarchical clustering. In \cite{vincentNIPS}, they show that average linkage is a $\frac{1}{2}$-approximation for this problem and they propose a top-down approach based on locally densest cut achieving a $(\frac{2}{3}-\epsilon)$-approximation in time $\tilde{O}\left(\frac{n^2(n+m)}{\epsilon}\right)$. Notably, when $\epsilon$ gets small the running time blows up. 

Here, we prove that the most natural randomized algorithm for this problem, i.e. recursive random cutting, is a $\frac{2}{3}$-approximation with expected running time $O(n\log n)$. We further derandomize this algorithm to get a simple deterministic local-search style $\frac{2}{3}$-approximation algorithm. 

If we also have structural constraints for the dissimilarity-HC, we show that the existing approaches fail. In fact we show that they lead to an $\Omega(n)$-approximation factor due to the lack of ``exploration'' (e.g. recursive densest cut). We then show that recursive random cutting is robust to adding user constraints, and indeed it preserves a constant approximation factor when there are, roughly speaking, constantly many user constraints. 

\paragraph{Randomized $\mathbf{\tfrac{2}{3}}$-approximation.} Consider the most natural randomized algorithm for hierarchical clustering, i.e. recursively partition each cluster into two, where each point in the current cluster independently flips an unbiased coin and based on the outcome, it is put in one of the two parts. 
\begin{theorem}
\label{thm:recursive-random}
\RRC  is a $\frac{2}{3}$-approximation for maximizing dissimilarity-HC objective.
\end{theorem}
\begin{proof}[Proof sketch.]
An alternative view of Dasgupta's objective is to divide the reward of the clustering tree between all possible triples $\{i,j,k\}$, where $(i,j)\in E$ and $k$ is another point (possibly equal to $i$ or $j$). Now, in any hierarchical clustering tree, if at the moment right before $i$ and $j$ become separated the vertex $k$ has still been in the same cluster as $\{i,j\}$, then this triple contributes $w_{ij}$ to the objective function. We claim this event happens with probability exactly ${\frac{2}{3}}$. To see this, consider an infinite independent sequence of coin flips for $i$, $j$, and $k$. Without loss of generality, condition on $i$'s sequence to be all heads. The aforementioned event happens only if $j$'s first tales in its sequence happens no later than $k$'s first tales in its sequence. This happens with probability $\sum_{i\geq 1} \frac{1}{2}(\frac{1}{4})^{i-1}=\frac{2}{3}$. Therefore, the algorithm gets the total reward $\frac{2n}{3}\sum_{(i,j)\in E}w_{ij}$ in expectation. Moreover, the total reward of any hierarchical clustering is upper-bounded by $n\sum_{(i,j)\in E}w_{ij}$, which completes the proof of the $\frac{2}{3}$-approximation. 
\end{proof}

\begin{remark} This algorithm runs in time $O(n\log n)$ in expectation, due to the fact that the binary clustering tree has expected depth $O(\log n)$ (see for example \cite{CLRS2009}) and at each level we only perform $n$ operations. 
\end{remark}
We now derandomize the recursive random cutting algorithm using the \emph{method of conditional expectations}. At every recursion, we go over the points in the current cluster one by one, and decide whether to put them in the ``left'' partition or ``right'' partition for the next recursion. Once we make a decision for a point, we fix that point and go to the next one. Roughly speaking, these local improvements can be done in polynomial time, which will result in a simple local-search style deterministic algorithm.
\begin{theorem}
\label{thm:derandom}
There is a deterministic local-search style $\frac{2}{3}$-approximation algorithm for maximizing dissimilarity-HC objective that runs in time $O(n^2(n+m))$.
\end{theorem}

\paragraph{Maximizing the Objective with User Constraints} From a practical point of view, one can think of many settings in which the output of the hierarchical clustering algorithm should satisfy user-defined hard constraints. Now, combining the new perspective of maximizing Dasgupta's objective with this practical consideration  raises a natural question: which algorithms are robust to adding user constraints, in the sense that a simple variation of these algorithms still achieve a decent approximation factor? 

\paragraph{$\bullet$ Failure of ``Non-exploring'' Approaches.}

Surprisingly enough, there are convincing reasons that adapting existing algorithms for maximizing Dasgupta's objective (e.g. those proposed in ~\cite{vincentSODA}) to handle user constraints is either challenging or hopeless. First, bottom-up algorithms, e.g. average-linkage, fail to output a feasible outcome if they only consider each constraint separately and not all the constraints jointly (as we saw in \hyperref[fig:fail]{Figure~\ref{fig:fail}}). Second, maybe more surprisingly, the natural extension of (locally) \texttt{Recursive-Densest-Cut}\footnote{While a locally densest cut can be found in poly-time, desnest cut is NP-hard, making our negative result stronger.} algorithm proposed in~\cite{vincentSODA} to handle user constraints performs poorly in the worst-case, even when we have only one constraint. \texttt{Recursive-Densest-Cut} proceeds by repeatedly picking the cut that has maximum density, i.e. ${\argmax}_{S\subseteq V}\frac{w(S,\bar{S})}{\lvert S\rvert \cdot\lvert \bar{S}\rvert}$ and making two clusters. To handle the user constraints, we run it  recursively on the supergraph generated by the constraints, similar to the approach in \hyperref[sec:CSC]{Section~\ref{sec:CSC}}. Note that once the algorithm resolves a triplet constraint, it also breaks its corresponding supernode. 

Now consider the following example in \hyperref[fig:counter]{Figure~\ref{fig:counter}}, in which there is just one triplet constraint ab$|$c.  The weight $W$ should be thought of as large and $\epsilon$ as small. By choosing appropriate weights on the edges of the clique $K_n$, we can fool the algorithm into cutting the dense parts in the clique, without ever resolving the $ab|c$ constraint until it is too late. The algorithm gets a gain of $O(n^3+W)$ whereas \texttt{OPT} gets $\Omega(nW)$ by starting with the removal of the edge $(b,c)$ and then removing $(a,b)$, thus enjoying a gain of $\approx nW$.  
\begin{figure}[ht]
\vskip -0.1in
\label{fig:lowerbound}
\begin{center}
\centerline{\includegraphics[scale=0.6]{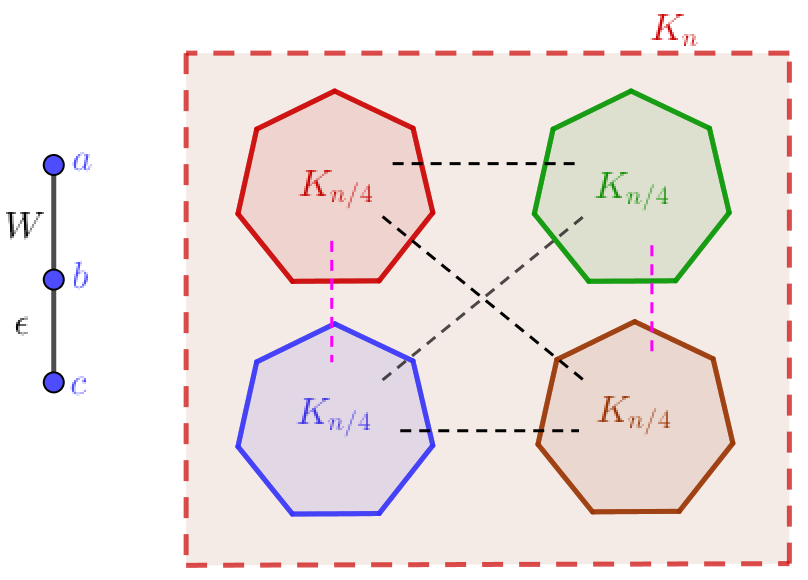}} \label{fig:counter}
\caption{$\Omega(n)$-approximation lower bound instance for the constrained \texttt{Recursive-Densest-Cut} algorithm.}
\end{center}
\vskip -0.4in
\end{figure}

\paragraph{$\bullet$ Constrained Recursive Random Cutting.} The  example in \hyperref[fig:lowerbound]{Figure~\ref{fig:lowerbound}}, although a bit pathological, suggests that a meaningful algorithm for this problem should explore cutting low-weight edges that might lead to resolving constraints, maybe randomly, with the hope of unlocking rewarding edges that were hidden before this exploration. 

Formally, our approach is showing that the natural extension of recursive random cutting for the constrained problem, i.e. by running it on the supergraph generated by constraints and unpacking supernodes as we resolve the constraints (in a similar fashion to \texttt{CSC}), achieves a constant factor approximation when the constraints have bounded \emph{dependency}. In the remaining of this section, we define an appropriate notion of dependency between the constraints, under the name of \emph{dependency measure} and analyze the approximation factor of constrained recursive random cutting (\CRRC) based on this notion.
 
Suppose we are given an instance of hierarchical clustering with triplet constraints $\{c_1,\ldots,c_k\}$, where  $c_i=x^{i}|y^{i}z^{i}, \forall i\in[k]$. For any triplet constraint $c_i$, lets call the pair $\{y^i,z^i\}$ the \emph{base}, and $z^i$ the \emph{key} of the constraint. We first partition our constraints into equivalence classes $\cons_1,\ldots,\cons_N$, where $\cons_i\subseteq \{c_1,\ldots,c_k\}$.  For every $i,j$, the constraints $c_i$ and $c_j$ belong to the same class $\mathcal{C}$ if they share the same base (see \hyperref[fig:class]{Figure~\ref{fig:class}}).
\begin{figure}[t]
        \centering
       \includegraphics[scale=0.6]{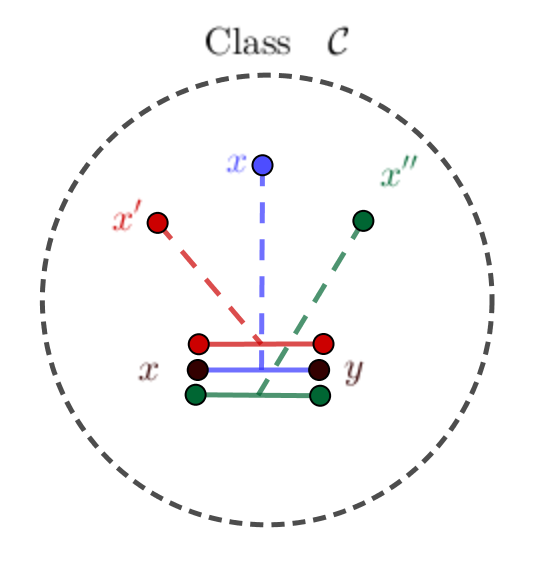}
        \caption{Description of a class $\cons$ with base $\{x,y\}$.\label{fig:class} }
         \vspace{-0.35cm}
\end{figure}
   
\begin{definition}[\textbf{Dependency digraph}] 
\label{def:depdigraph}
The Dependency digraph is a directed graph with vertex set $\{\cons_1,\ldots,\cons_L\}$. For every $i,j$, there is a directed edge $\cons_i\rightarrow \cons_j$ if $\exists~c=x|yz, c'=x'|y'z'$, such that $c\in \cons_i, c'\in\cons_j$, and either $\{x,z\}=\{y',z'\}$ or $\{x,y\}=\{y',z'\}$ (see \hyperref[fig:depdi]{Figure~\ref{fig:depdi}}).
\end{definition}
    
\begin{figure}
        \centering
       \includegraphics[width=0.8\textwidth]{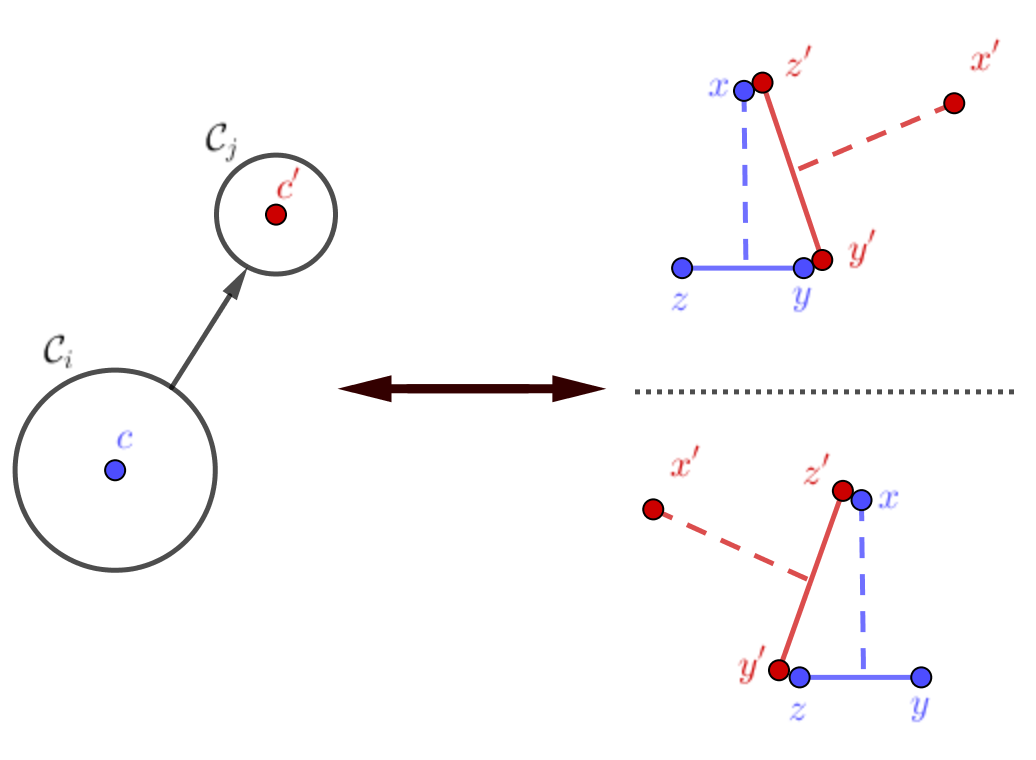}
        \caption{Classes $\{\cons_{i},\cons_j\}$, and two situations for having $\cons_i\rightarrow\cons_j$.\label{fig:depdi}}
         \vspace{-0.35cm}
\end{figure}

The dependency digraph captures how groups of constraints  impact each other. Formally, the existence of the edge $\cons_i\rightarrow \cons_j$ implies that all the constraints in $\cons_j$ should be resolved before one can separate the two endpoints of the (common) base edge of the constraints in $\cons_i$.
\begin{remark}
If the constraints $\{c_1,\ldots,c_k\}$ are feasible, i.e. there exists a hierarchical clustering that can respect all  the constraints, the dependency digraph is clearly acyclic.
\end{remark}
\begin{definition}[\textbf{Layered dependency subgraph}]
\label{def:depsub}
Given any class $\cons$, the \emph{layered dependency subgraph of $\cons$} is the induced subgraph in the dependency digraph by all the classes that are reachable from $\cons$. Moreover, the vertex set of this subgraph can be partitioned into layers $\{\layer_0,\layer_1,\ldots,\layer_L\}$, where $L$ is the maximum length of any directed path leaving $\cons$ and $\layer_l$ is a subset of classes where the length of the longest path from $\cons$ to each of them is exactly equal to $l$ (see \hyperref[fig:layers]{Figure~\ref{fig:layers}}).
\end{definition}

\begin{figure}[thb]
    \centering
       \includegraphics[width=0.8\textwidth]{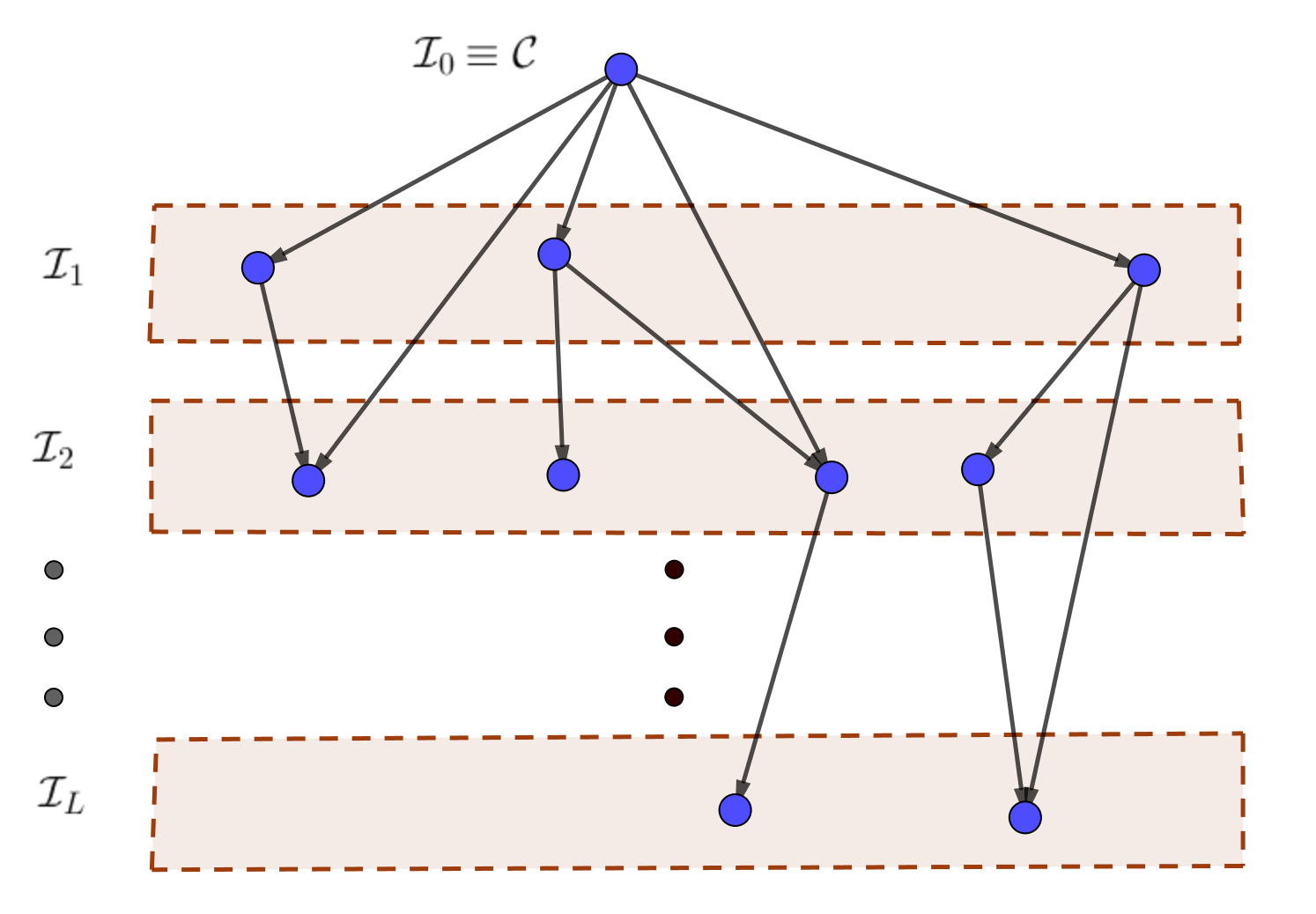}
   \caption{Layered dependency subgraph of class $\cons$.\label{fig:layers}}
\end{figure}

We are now ready to define a crisp quantity for every dependency graph. This will later help us give a more meaningful and refined \textit{beyond-worst-case} guarantee for the approximation factor of the \CRRC algorithm.
\begin{definition}[\textbf{Dependency measure}]
\label{def:dep-mes}
Given any class $\cons$, the \emph{dependency measure of $\cons$} is defined as
$$
\dm{\cons}\triangleq\prod_{l=0}^L(1+\sum_{\cons'\in \layer_l} \lvert\cons'\rvert),
$$
where $\layer_0,\ldots,\layer_L$ are the layers of the dependency subgraph of $\cons$, as in \hyperref[def:depsub]{Definition~\ref{def:depsub}}. Moreover, the dependency measure of a set of constraints $\dmc{\{c_1,\ldots,c_k\}}$ is defined as $\max_{\cons}\dm{\cons}$, where the maximum is taken over all the classes generated by $\{c_1,\ldots,c_k\}$.
\end{definition}
Intuitively speaking, the notion of the dependency measure quantitatively expresses how ``deeply'' the base of a constraint is \emph{protected} by the other constraints, i.e. how many constraints need to be resolved first before the base of a particular constraint is unpacked and the \CRRC algorithm can enjoy its weight. This intuition is formalized through the following theorem, whose proof is deferred to the supplementary materials. 
 \begin{theorem} 
\label{thm:constrained-random} 
 The constrained recursive random cutting (\CRRC) algorithm is an $\alpha$-approximation algorithm for maximizing dissimilarity-HC objective objective given a set of feasible constraints $\{c_1,\ldots,c_k\}$, where 
 $$\alpha=\frac{2(1-k/n)}{3\cdot\dmc{\{c_1,\ldots,c_k\}}}\le\frac{2(1-k/n)}{3\cdot\max_{\cons}\dm{\cons}}$$
 \end{theorem}
 
 \begin{corollary}
\CRRC is an $O(1)$-approximation for maximizing dissimilarity-HC objective, given feasible constraints of constant dependency measure.
 \end{corollary}

\section{Conclusion}
We studied the problem of hierarchical clustering when we have structural constraints on the feasible hierarchies. We followed the optimization viewpoint that was recently developed in~\cite{dasguptaSTOC,vincentSODA} and we analyzed two natural top-down algorithms giving provable approximation guarantees. In the case where the constraints are infeasible, we proposed and analyzed a regularized version of the HC objective by using the hypergraph version of the sparsest cut problem. Finally, we also explored a variation of Dasgupta's objective and improved upon previous techniques, both in the unconstrained and in the constrained setting.


\section*{Acknowledgements}
Vaggos Chatziafratis was partially supported by ONR grant N00014-17-1-2562. Rad Niazadeh was supported by Stanford Motwani fellowship. Moses Charikar was supported by NSF grant CCF-1617577 and a Simons Investigator Award. We would also like to thank Leo Keselman, Aditi Raghunathan and Yang Yuan for providing comments on an earlier draft of the paper. We also thank the anonymous reviewers for their helpful comments and suggestions.

\bibliographystyle{plainnat}
\bibliography{refs}
\appendix

\section{Supplementary Materials}
\subsection{Missing proofs and discussion in Section~\ref{sec:CSC}}
\begin{proof}[Proof of Proposition~\ref{prop:convert}]
For nodes $u,v \in T$, let $P(u)$ denote the parent of $u$ in the tree and $\textrm{LCA}(u,v$) denote the lowest common ancestor of $u,v$. For a leaf node $l_i, i\in [k]$, we say that its label is $l_i$, whereas for an internal node of $T$, we say that its label is the label of any of its two children. As long as there are any two nodes $a,b$ that are siblings (i.e. $P(a)\equiv P(b)$), we create a constraint $ab|c$ where $c$ is the label of the second child of $P(P(a))$. We delete leaves $a,b$ from the tree and repeat until there are fewer than $3$ leaves left. To see why the above procedure will only create at most $k$ constraints, notice that every time a new constraint is created, we delete two nodes of the given tree $T$. Since $T$ has $k$ leaves and is binary, it can have at most $2k-1$ nodes in total. It follows that we create at most $\tfrac{2k-1}{2}< k$ triplet constraints. For the equivalence between the constraints imposed by $T$ and the created triplet constraints, observe that \textit{all} triplet constraints we create are explicitly imposed by the given tree (since we only create constraints for two leaves that are siblings) and that for any three datapoints $a,b,c \in T$ with LCA($a,c$)=LCA($b,c$), our set of triplet constraints will indeed imply $ab|c$, because LCA($a,b$) appears further down the tree than LCA($a,c$) and hence $a,b$ become siblings before $a,c$ or $b,c$.
\end{proof}

\begin{proof}[Proof of Fact~\ref{fact:dec} from \cite{vaggosSODA}]
We will measure the contribution of an edge $e=(u, v) \in E$ to the RHS and to the LHS. Suppose that $r$ denotes the size of the \textit{minimal} cluster in $\opt$ that contains both $u$ and $v$. Then the contribution of the edge $e=(u, v)$ to the LHS is by definition $r\cdot w_{e}$. On the other hand, $(u, v)\in  \opt(t), \forall t \in \{0, . . . , r-1\}$. Hence the contribution to the RHS is also $r\cdot w_{e}$.
\end{proof}

\begin{proof}[Proof of Fact~\ref{fact:scale} from \cite{vaggosSODA}]
We rewrite $\opt$ using the fact that $$w(\opt(t))\ge0$$ at every level $t\in[n]$: 
\begin{align*}
6k\cdot \opt&=6k\sum_{t=0}^n w(\opt(t))\\
&=6k(w(\opt(0))+\dots+w(\opt(n)))\\
&\ge 6k(w(\opt(0))+\dots+ w(\opt(\lfloor\tfrac{n}{6k}\rfloor)))\\
&= \sum_{t=0}^nw({\opt(\lfloor\tfrac{t}{6k}\rfloor))}
\end{align*}
\end{proof}

\begin{proof}[Proof of Lemma~\ref{lem:combine}]
By using the previous lemma we have:
\begin{align*}
\csc=\sum_{A} r_A w(B_1,B_2)\le 
\le O(\alpha_n)\sum_A s_A w(\opt(\lfloor\tfrac{r_A}{6k_A}\rfloor)\cap A)
\end{align*}
Observe that $w(\opt(t))$ is a decreasing function of $t$, since as $t$ decreases, more and more edges are getting cut. Hence we can write:
\begin{align*}
\sum_A s_A\cdot w(\opt(\lfloor\tfrac{r_A}{6k}\rfloor)\cap A)\le \sum_A\sum_{t=r_A-s_A+1}^{r_A}w(\opt(\lfloor\tfrac{r_A}{6k_A}\rfloor)\cap A)
\end{align*}
To conclude with the proof of the first part all that remains to be shown is that:
\[
\sum_A\sum_{t=r_A-s_A+1}^{r_A}w(\opt(\lfloor\tfrac{t}{6k_A}\rfloor)\cap A)\le \sum_{t=0}^nw(\opt(\lfloor\tfrac{t}{6k}\rfloor))
\]
To see why this is true consider the clusters $A$ with a contribution to the LHS. We have that $r_A-s_A + 1 \le t \le r_A$, hence $|B_2| < t$ meaning that $A$ is
a \textit{minimal} cluster of size $|A| \ge t > |B_2| \ge |B_1|$, i.e. if both $A$'s children are of size less than $t$,
then this cluster $A$ contributes such a term. The set of all such $A$ form a disjoint partition of $V$
because of the definition for minimality (in order for them to overlap in the hierarchical clustering,
one of them needs to be ancestor of the other and this cannot happen because of minimality).
Since $\opt(\lfloor\tfrac{t}{6k}\rfloor)\cap A$ for all such $A$ forms a disjoint partition of $\opt(\lfloor\tfrac{t}{6k}\rfloor)$, the claim follows
by summing up over all $t$.

Note that so far our analysis handles clusters $A$ with size $r_A\ge 6k$. However, for clusters with smaller size $r_A<6k$ we can get away by using a crude bound for bounding the total cost and still not affecting the approximation guarantee that will be dominated by $O(k\alpha_n)$:
\[
\sum_{|A|<6k}r_Aw(B_1,B_2)< 6k\cdot \sum_{ij\in E}w_{ij}=6k\cdot \opt(1)\le 6k\cdot \opt\qedhere
\]
\end{proof}
\begin{theorem}[The divisive algorithm using balanced cut] 
Given a weighted graph $G(V,E,w)$ with $k$ triplet constraints $ab|c$ for $a,b,c\in V$, the constrained recursive balanced cut algorithm (same as $\csc$, but using balanced cut instead of sparsest cut) outputs a HC respecting \textit{all} triplet constraints and achieves an $O(k\alpha_n)$-approximation for the HC objective \hyperref[obj]{(\ref{obj})}.
\end{theorem}
\begin{proof}
It is not hard to show that one can use access to balanced cut rather than sparsest cut and achieve the same approximation factor by the recursive balanced cut algorithm. 

We will follow the same notation as in the sparsest cut analysis and we will use some of the facts and inequalities we previously
proved about \texttt{OPT(t)}. Again, for a cluster $A$ of size $r$, the important observation is that the partition $A_1,\dots, A_{l}$ (at the end, we will again choose $l=6k_A$) induced inside the cluster $A$ by \texttt{OPT}$(\tfrac{r}{l})$ can be separated into two groups, let’s
say $(C_1, C_2)$ such that $r/3\le |C_1|, |C_2| \le 2r/3$. In other words we can demonstrate a Balanced
Cut with ratio $\tfrac13:\tfrac23$ for the cluster $A$. Since we cut fewer edges when creating $C_1, C_2$ compared to the partitioning of \texttt{OPT}$(\tfrac{r}{l})$:
\[
w(C_1,C_2)\le w(\opt(\lfloor\tfrac{r}{l}\rfloor)\cap A)
\]

By the fact we used an $\alpha_n$-approximation to balanced cut we can get the following inequality (similarly to \hyperref[lem:mainproof]{Lemma~\ref{lem:mainproof}}):
\[ 
r\cdot w(C_1,C_2)\le O(\alpha_n) \cdot s\cdot w(\opt(\lfloor\tfrac{r}{l}\rfloor)\cap A)
\]
Finally, we have to sum up over all the clusters $A$ (now in the summation we should write $r_A, s_A$ instead
of just $r, s$, since there is dependence in $A$) produced by the constrained recursive balanced cut algorithm for Hierarchical
Clustering and we get that we can approximate the
HC objective function up to $O(k\alpha_n)$.
\end{proof}

\begin{remark}Using balanced-cut can be useful for two reasons.
First, the runtime of sparsest and balanced cut on a graph with $n$ nodes and $m$ edges
are $\tilde{O}(m + n^{1+\epsilon})$. When run recursively however as in our case, taking recursive
sparsest cuts might be worse off by a factor of $n$ (in case of unbalanced splits at every step) in the
worst case. However, recursive balanced cut is still $\tilde{O}(m + n^{1+\epsilon})$. Second, it is known that an $\alpha$-approximation for the sparsest cut yields
an $O(\alpha)$-approximation for balanced cut, but not the
other way. This gives more flexibility to the balanced cut algorithm, and there is a chance it can achieve a better approximation factor (although we don't study it further in this paper).
\end{remark}

\subsection{Missing proofs in Section~\ref{sec:HRSC}}
\begin{proof}[Proof sketch of Proposition~\ref{prop:HSC}]
Here the main obstacle is similar to the one we handled when proving Theorem~(\ref{th:main}): for a given cluster $A$ created by the \texttt{R-HSC} algorithm, different constraints are, in general, active compared to the \opt\  decomposition for this cluster $A$. Note of course, that \opt\ itself will not respect all constraints, but because we don't know which constraints are active for \opt, we still need to use a charging argument to low levels of \opt. Observe that here we are allowed to cut an edge $ab$ even if we had the $ab|c$ constraint (incurring the corresponding cost $c_{ab|c}$), however we cannot possibly hope to charge this to the \opt\ solution, as \opt, for all we know, may have respected this constraint. In the analysis, we crucially use a merging procedure between sub-clusters of $A$ having active constraints between them and this allows us to compare the cost of our \texttt{R-HSC} with the cost of \opt\ .
\end{proof}

\begin{proof} [3-hyperedges to triangles for general weights ] Even though the general reduction presented in Section~\ref{sec:HRSC} (Figure~\ref{fig:reduction}) to transform a 3-hyperedge to a triangle is valid even for general instances of HSC with 3-hyperedges and arbitrary weights, the reduced sparsest cut problem may have negative weights, e.g. when $w_{bc|a}+w_{ac|b}<w_{ab|c}$. To the best of our knowledge, sparsest cut with negative weights has not been studied. Notice however that if the original weights $w_{bc|a},w_{ac|b},w_{ab|c}$ satisfy the triangle inequality (or as a special case, if two of them are zero which is usually the case when we have a triplet constraints), then we can actually solve (approximately) the HSC instance, as the sparsest cut instance will only have non-negative weights.
\end{proof}

\subsection{Missing proofs in Section~\ref{sec:random}}
\begin{proof}[Proof of Theorem~\ref{thm:recursive-random}]
We start by looking at the objective value of any algorithm as the summation of contributions of different triples $i,j$ and $k$ to the objective, where $(i,j)\in E$ and $k$ is some other point (possibly equal to $i$ or $j$). 
\begin{align*}
\OBJ=\sum_{(i,j)\in E}w_{ij}\lvert{T_{ij}}\rvert=\sum_{(i,j)\in E,k\in V}w_{ij}\mathbf{1}\{\textrm{$k\in \textrm{leaves}(T_{ij})$}\}
= \sum_{(i,j)\in E}\sum_{k\in V} \Ycal_{i,j,k},
\end{align*}
where random variable $\Ycal_{i,j,k}$ denotes the contribution of the edge $(i,j)$ and vertex $k$ to the objective value. The vertex $k$ is a leaf of $T_{ij}$ if and only if right before the time that $i$ and $j$ gets separated $k$ is still in the same cluster as $i$ and $j$. 
Therefore, 
$$\Ycal_{i,j,k}=w_{ij}\mathbf{1}\{\textrm{$i$ separates from $k$ no earlier than $j$ }\}$$

We now show that $\Ex{\Ycal_{i,j,k}}=\frac{2}{3}w_{ij}$. Given this, the expected objective value of  recursive random cutting algorithm will be at least $\frac{2n}{3}\sum_{(i,j)\in E}w_{ij}$. Moreover, the objective value of  the optimal hierarchical clustering, i.e. maximizer of the Dasgupta's objective, is no more than $n\sum_{(i,j)\in E}w_{ij}$, and we conclude that  recursive random cutting is a $\frac{2}{3}$-approximation. To see why  $\Ex{\Ycal_{i,j,k}}=\frac{2}{3}w_{ij}$,  think of randomized cutting as flipping an independent unbiased coin for each vertex, and then deciding on which side of the cut this vertex belongs to based on the outcome of its coin. Look at the sequence of the coin flips of $i$, $j$ and $k$. Our goal is to find the probability of the event that for the first time that $i$ and $j$ sequences are not matched, still $i$'s sequence and $k$'s sequence are matched up to this point, or still $j$'s sequence and $k$'s sequence are matched up to this. The probability of each of these events is equal to $\frac{1}{3}$. To see this for the first event, suppose $i$'s sequence is all heads ($H$). We then need the pair of coin flips of $(j,k)$ to be a sequence of $(H,H)$'s ending with a $(T,H)$, and this happens with probability $\sum_{i\geq 1} (\frac{1}{4})^i=\frac{1}{3}$. The probability of the second event is similarly calculated. Now, these events are disjoint. Hence, the probability that $i$ is separated from $k$ no earlier than $j$ is exactly $\frac{2}{3}$, as desired.
\end{proof}

\begin{proof}[Proof of Theorem~\ref{thm:derandom}] We derandomize the recursive random cutting algorithm  using the \emph{method of conditional expectations}. At every recursion, we go over the points in the current cluster one by one, and decide whether to put them in the ``left'' partition or ``right'' partition for the next recursion. Once we make a decision for a point, we fix that point and go to the next one. Now suppose for a cluster $C$ we have already fixed points $S\subseteq C$, and now we want to make a decision for $i\in C\setminus S$. The reward of assigning to left(right) partition is now defined as the expected value of recursive random cutting restricted to $C$, when the points in $S$ are fixed (i.e. it is already decided which points in $S$ are going to the left partition and which ones are going to the right partition), $i$ goes to the left(right) partition and $j\in C\setminus (\{i\}\cup S)$ are randomly assigned to either the left or right. Note that these two rewards (or the difference of the two rewards) can be calculated \emph{exactly} in polynomial time by considering all triples consisting of an edge and another vertex, and then calculating the probability that this triple contributes to the objective function (this is similar to the proof of Theorem~\ref{thm:recursive-random}, and we omit the details for brevity here). Because we know the randomized assignment of $i$ gives a $\frac{2}{3}$-approximation (Theorem~\ref{thm:recursive-random}), we conclude that assigning to the better of left or right partition for every vertex will remain to be at least a $\frac{2}{3}$-approximation. For running time, we have at most $n$ clusters to investigate. Moreover, a careful counting argument shows that the total number of operations required to calculate the differences of the rewards of assigning to left and right partitions for all vertices is at most $n(n+2m)$. Hence, the running time is bounded by $O(n^2(n+m))$.
\end{proof}
\begin{proof}[Proof sketch of Theorem~\ref{thm:constrained-random}.] Before starting to prove the theorem, we prove the following simple lemma.
\begin{lemma}
\label{lem:ind}
There is no edge between any two classes in the same layer $\layer_l$.
\end{lemma}
\begin{proof}[Proof of Lemma~\ref{lem:ind}]
 If such an edge exists, then there is a path of length $l+1$ from $\cons$ to a class in $\layer_l$, a contradiction.
\end{proof}

Now, similar to the proof of Theorem~\ref{thm:recursive-random}, we consider every triple $\{x,y,z\}$, where $(x,y)\in E$ and $z$ is another point , but this time we only consider $z$'s that are not involved in any triplet constraint (there are at least $n-k$ such points). We claim with probability at least $\frac{2}{3\cdot\dmc{\{c_1,\ldots,c_k\}}}$ the supernode containing $z$ is still in the same cluster as supernodes containing $x$ and $y$ right before $x$ and $y$ gets separated. By summing over all such triples, we show that the algorithm gets a gain of at least $\frac{2(n-k)}{3\cdot\dmc{\{c_1,\ldots,c_k\}}}\sum_{(x,y)\in E}w_{xy}$, which proves the $\alpha$-approximation as the optimal clustering has a reward bounded by $n\sum_{(x,y)\in E}w_{xy}$. 
 
 To prove the claim, if $(x,y)$ is not the base of any triplet constraint then a similar argument as in the proof of Theorem~\ref{thm:recursive-random} shows the desired probability is exactly $\frac{2}{3}$ (with a slight adaptation, i.e. by looking at the coin sequences of supernodes containing $x$ and $y$, which are going to be disjoint in this case at all iterations, and the coin sequence of $z$). Now suppose $(x,y)$ is the base of any constraint $c$ and suppose $c$ belongs to a class $\cons$. Consider the layered dependency subgraph of $\cons$ as in Definition~\ref{def:depsub} and let the layers to be $\layer_0,\ldots,\layer_L$. In order for $z$ to be in the same cluster as $x$ and $y$ when they get separated, a chain of $L+1$ independent events needs to happen. These events are defined inductively; for the first event, consider the coin sequence of  $z$, coin sequence of (the supernode containing all the bases of) constraints in $\cup_{l=0}^L \layer_l$ and coin sequences of all the keys of constraints in $\layer_L$ (there are $\sum_{\cons'\in \layer_L}\lvert \cons'\rvert$ of them). Without loss of generality, suppose the coin sequence of (the supernode containing) $\cup_{l=0}^L \layer_l$  is all heads. Now the event happens only if at the time $z$ flips its first tales all keys of $\layer_L$ have already flipped at least one tales. Conditioned on this event happening, all the constraints in $\layer_L$ will be resolved and $z$ remains in the same cluster as $x$ and $y$. Now, remove $\layer_L$ from the dependency subgraph and repeat the same process to define the events $2,\ldots,L$ in a similar fashion. For the $l^{\textrm{th}}$ event to happen, we need to look at $1+\sum_{\cons'\in \layer_L}\lvert \cons'\rvert$ number of i.i.d. symmetric geometric random variable, and calculate the probability that first of them is no smaller than the rest. This event happens with a probability at least $\left(1+\sum_{\cons'\in \layer_L}\lvert \cons'\rvert\right)^{-1}$. Moreover the events are independent, as there is no edge between any two classes in $\layer_l$ for $l\in[L]$, and different classes have different keys. After these $L$ events, the final event that needs to happen is when all the constraints are unlocked, and $z$ needs to remain in the same cluster as $x$ and $y$ at the time they get separated. This event happens with probability $\frac{2}{3}$. Multiplying all of these probabilities due to independence implies the desired approximation factor.
 \end{proof}

\section{Experiments}
\label{sec:appendix-zoo}
The purpose of this section is to present the benefits of incorporating triplet constraints when performing Hierarchical Clustering. We will focus on real data using the Zoo dataset~(\cite{Zoo}) for a taxonomy application. We demonstrate that using our approach, the performance of simple recursive spectral clustering algorithms can be improved by approximately $9\%$ as measured by the Dasgupta's Hierarchical Clustering cost function (\hyperref[obj]{\ref{obj}}). More specifically:

\begin{itemize}

\item \emph{The Zoo dataset}: It contains 100 animals forming 7 different categories (e.g. mammals, amphibians etc.). The features of each animal are provided by a 16-dimensional vector containing information such as if the animal has hair or feathers etc.

\item \emph{Evaluation method}: Given the feature vectors, we can create a similarity matrix $M(\cdot,\cdot)$ indexed by the labels of the animals. We choose the widely used cosine similarity to create $M$.

\item \emph{Algorithms}: We use a simple implementation of spectral clustering based on the second eigenvector of the normalized Laplacian of $M$. By applying the spectral clustering algorithm once, we can create two clusters; by applying it recursively we can create a complete hierarchical decomposition, which is ultimately the output of the HC algorithm.

\item \emph{Baseline comparison}: Since triplet constraints are especially useful when there is noisy information (i.e. noisy features), we simulate this situation by hiding some of the features of our Zoo dataset. Specifically, when we want to find the target HC tree $T^*$, we use the full 16-dimensional feature vectors, but for the comparison between the unconstrained and the constrained HC algorithms we will use a noisy version of the feature vectors which consists of only the first 10 coordinates from every vector.

In more detail, the first step in our experiments is to evaluate the cost of the target clustering $T^*$. For this, we use the full feature vectors and perform repeated spectral clustering to get a hierarchical decomposition (without incorporating any constraints). We call this cost $\texttt{OPT}$.

The second step is to perform unconstrained HC but with noisy information, i.e. to run the spectral clustering algorithm repeatedly on the 10-dimensional feature vectors (again without taking into account any triplet constraints). This will output a hierarchical tree that has cost in terms of the Dasgupta's HC cost $\texttt{Unconstrained\_Noisy\_Cost}.$\footnote{The cost of the trees are always evaluated using the actual similarities obtained from the full feature vectors.}

The final step is to choose some structural constraints (that are valid in $T^*$)\footnote{Here we chose triplet constraints that will induce the same first cut as $T^*$ and required no constraints after that. This corresponds to a high-level separation of the animals, for example to those that are ``land'' animals versus those that are ``water'' animals.} and perform again HC with noisy information. We again use the 10-dimensional feature vectors but the spectral clustering algorithm is allowed only cuts that do not violate any of the given structural constraints. Repeating  until we get a decomposition gives us the final output which will have cost in terms of the Dasgupta's HC cost $\texttt{Constrained\_Noisy\_Cost}$. 
\end{itemize}

The first main result of our experimental evaluation is that the $\texttt{Constrained\_Noisy\_Cost}$ is surprisingly close to $\texttt{OPT}$, even though to get the $\texttt{Constrained\_Noisy\_Cost}$ the features used were noisy and the second main result is that incorporating the structural constraints yields $\approx 9\%$~improvement over the noisy unconstrained version of HC with cost $\texttt{Unconstrained\_Noisy\_Cost}$. Now that we have presented the experimental set-up, we can proceed by describing our results and final observations in greater depth.

\subsection{Experimental Results}

We ran our experiments for $20,50,80$ and $100$ animals from the Zoo dataset and for the evaluation of the $\%$ improvement in terms of the Dasgupta's HC cost (\hyperref[obj]{\ref{obj}}), we used the following formula:
\[
\frac{\texttt{Unconstrained\_Noisy\_Cost}-\texttt{Constrained\_Noisy\_Cost}}{\texttt{OPT}}
 \]
 
 The improvements obtained due to the constrained version are presented in~\hyperref[tab:exp]{Table~\ref{tab:exp}}.

\begin{center}

\begin{table}

 \begin{tabular}{|c|||c|c|c|||c|} 

  \hline
 $\#$animals & $\texttt{OPT}$ & $\texttt{Unconstrained\_Noisy\_Cost}$ & $\texttt{Constrained\_Noisy\_Cost}$ & $\%$ Improvement  \\ [0.5ex] 
 \hline\hline
 20 & 1137 & 1286 & 1142 & 12.63\\ 
 \hline
 50 & 23088 & 25216 & 23443 &7.68\\
 \hline
 80 & 89256 & 99211 & 90419 &9.85\\
 \hline
 100 & 171290 & 190205 & 173499 &9.75\\
 \hline
\end{tabular}

\caption{Results obtained for the Zoo dataset. The improvement corresponds to the lower cost of the output HC tree after incorporating structural constraints, even in the presence of noisy features. Observe that in all cases the performance of $\texttt{Constrained\_Noisy\_Cost}$ is extremely close to the $\texttt{OPT}$ cost.}\label{tab:exp}
\end{table}
\end{center}

Some observations regarding the structural constraints are the following:

\begin{itemize}

\item When we add triplet constraints to the input as advice for the algorithm, it is crucial for the triplet constraints to actually be useful. ``Easy'' constraints that are readily implied by the similarity scores will have no extra use and will not lead to better solutions.
\item We also observed that having ``nested'' constraints can be really useful. Nested constraints can guide our algorithm to perform good cuts as they refer to a larger portion of the optimum tree $T^*$ (i.e. contiguous subtrees) rather than just different unrelated subtrees of it. The usefulness of the given constraints is correlated with the depth of the nested constraints and their accordance with the optimum tree $T^*$ based on Dasgupta's objective.
\item Furthermore, since most of the objective cost comes from the large initial clusters, we focused on the partitions that created large clusters and imposed triplet constraints that ensured good cuts in the beginning. Actually in some cases, just the first 3 or 4 cuts are enough to guarantee that we get $\approx 12\%$ improvement.

\item Finally, we conclude that just the number of the given triplet constraints may not constitute a good metric for their usefulness. For example, a large number of constraints referring to wildly different parts of $T^*$, may end up being much less useful than a smaller number of constraints guiding towards a good first cut.

\end{itemize}

\end{document}